\documentclass[10pt]{article}
\usepackage{amsmath,amsfonts,amssymb,mathrsfs}

\newtheorem{theorem}{Theorem}[section]
\newtheorem{lemma}{Lemma}[section]

\newtheorem{remark}{Remark}[section]

\newcommand{\bbr}{\mathbb R}
\def\charf {\mbox{{\text 1}\kern-.24em {\text l}}}

\newenvironment{proof}{\noindent {\it Proof.}}{\hfill$\square$\bigskip}

\begin{document}

\title{The Einstein-Boltzmann system and positivity}
\renewcommand{\baselinestretch}{1}
\author{Ho Lee and Alan D. Rendall\\
$\;$\\
Max-Planck-Intitut f\"{u}r Gravitationsphysik, Albert-Einstein-Institut,\\
Am M\"{u}hlenberg 1, 14476 Potsdam, Germany\\
\texttt{E-mail:ho.lee@aei.mpg.de, alan.rendall@aei.mpg.de}\\
}
\date{}

\maketitle

\begin{abstract}
The Einstein-Boltzmann system is studied, with particular attention to the 
non-negativity of the solution of the Boltzmann equation. A new 
parametrization of post-collisional momenta in general relativity
is introduced and then used to simplify the conditions on the collision
cross-section given by Bancel and Choquet-Bruhat in \cite{BCB73}. The
non-negativity of solutions of the Boltzmann equation on a given curved
spacetime has been studied by Bichteler \cite{B67} and Tadmon \cite{T10}. 
By examining to what extent the results of these authors apply in the 
framework of Bancel and Choquet-Bruhat, the non-negativity problem for the 
Einstein-Boltzmann system is resolved for a certain class of scattering
kernels. It is emphasized that it is a challenge to extend the existing
theory of the Cauchy problem for the Einstein-Boltzmann system so as to
include scattering kernels which are physically well-motivated.
\end{abstract}

\tableofcontents

\section{Introduction}
\setcounter{equation}{0}
To understand how the gravitational field and matter interact in general 
relativity it is necessary to specify a suitable matter model. One 
frequently used type of model comes from kinetic theory, where matter is 
considered as a collection of particles described statistically. The Vlasov 
and Boltzmann equations are the equations describing this type of model
and coupling to the gravitational field gives the Einstein-Vlasov and the 
Einstein-Boltzmann systems respectively. We refer to \cite{A11} 
and \cite{R} and their references for basic information about the 
Einstein-Vlasov system and the Einstein equations coupled to many different 
matter models. The present paper considers only the Einstein-Boltzmann system. 
Compared to the Einstein-Vlasov system, the Einstein-Boltzmann (EB) system
has not been studied much in the decades since local existence was proven
by Bancel and Choquet-Bruhat \cite{B73,BCB73}. Recently several further 
results have been obtained for the EB system. Mucha \cite{M981,M982} showed 
that the initial regularity conditions given in \cite{B73,BCB73} can be 
weakened. Noutchegueme and Dongo \cite{ND06} and Noutchegueme and 
Takou \cite{NT06} proved some global existence results for certain
classes of spacetimes with symmetries. The relativistic Boltzmann equation on 
a given curved spacetime was studied in \cite{B67,NDT05,T10}. As for the 
nonrelativistic or special relativistic Boltzmann equations, one can find 
plenty of references, and we only refer to \cite{CIP,dvv,G}.

The EB system is a system of equations where the Einstein equations and
the Boltzmann equation are coupled to each other. In harmonic coordinates, the 
Einstein equations take the form of the system of quasilinear wave equations 
\begin{equation}\label{1.1}
-\frac{1}{2}g^{\gamma\delta}\frac{\partial^2g^{\alpha\beta}}
{\partial x^\gamma\partial x^\delta}
+H^{\alpha\beta}=T^{\alpha\beta}-\frac12 (g^{\gamma\delta}T_{\gamma\delta})g^{\alpha\beta},
\end{equation}
for a Lorentzian metric $g_{\alpha\beta}$, where $H^{\alpha\beta}$ is a rational 
function of $g_{\alpha\beta}$ and 
$\partial_\gamma g_{\alpha\beta}$ with denominator a power of the determinant 
$|g|$, and $T^{\alpha\beta}$ is the stress-energy tensor of matter described by 
the Boltzmann equation. The Boltzmann equation describes the dynamics of a 
relativistic gas, the particles of which interact through binary collisions.
The unknown is the distribution function $f$ defined on the phase space,
which is the tangent bundle $T(M)$ of a four-dimensional manifold $M$
with a metric $g_{\alpha\beta}$. If we assume that all the particles considered 
have the same mass, then momentum is confined to a submanifold 
$P(M)\subset T(M)$, which is defined by $g_{\alpha\beta}p^\alpha p^\beta=-1$, when 
the mass is normalized to unity. This submanifold is called the mass 
hyperboloid. As a consequence of this assumption the distribution function 
turns out to be a function of seven variables. The Boltzmann equation takes 
the form
\begin{equation}\label{1.2}
p^\alpha\frac{\partial f}{\partial x^\alpha}-\Gamma^i_{\alpha\beta}p^\alpha p^\beta
\frac{\partial f}{\partial p^i}=Q(f,f),
\end{equation}
where $\Gamma^i_{\alpha\beta}$ are the Christoffel symbols of the metric 
$g_{\alpha\beta}$, and the collision operator $Q$,
\begin{equation}\label{1.3}
Q(f,f)=\int_{\bbr^3}\int_{S^2}S(x,p,q,\Omega)
\Big(f(p')f(q')-f(p)f(q)\Big)\,d\Omega\,\frac{|g|^{\frac{1}{2}}}{-q_0}\,dq,
\end{equation}
will be studied in detail in the next section.
Here, for simplicity we abbreviate $f(x,p')$ as $f(p')$, $f(x,q')$ as $f(q')$, and so on.
The stress-energy tensor $T^{\alpha\beta}$ is defined as
\begin{equation}\label{1.4}
T^{\alpha\beta}=\int_{\bbr^3}f(x,p)p^\alpha p^\beta\frac{|g|^{\frac{1}{2}}}{-p_0}\,dp,
\end{equation}
where the integration over $\bbr^3$ is understood as integration on the mass 
hyperboloid at $x$.
The equations \eqref{1.1}--\eqref{1.4} form the EB system \cite{Eh,St}.

In this paper, we are interested in the strong solutions of the EB system
constructed in \cite{BCB73} by Bancel and Choquet-Bruhat. In Section 2, the 
Boltzmann equation in a curved spacetime is studied and an alternative 
formulation using an orthonormal frame is explained. (Its relevance is
discussed in Section 4.) Returning to the equations written in a coordinate 
frame as in \cite{BCB73} we introduce a new parametrization of the 
post-collisional momenta. Unlike in the nonrelativistic case, the 
post-collisional momenta can usefully be parametrized in many different ways 
in the relativistic case \cite{GS91,S102,S11}. The parametrization in Section 
2.1 can be thought of as a generalization of that in \cite{GS91}. With this 
parametrization, we consider the original result of Bancel and Choquet-Bruhat, 
and simplify some conditions in their paper. The distribution function in the 
Boltzmann equation should be non-negative because of its interpretation as a 
number density. In mathematical terms this means that a solution $f$ of the 
Boltzmann equation arising from non-negative initial data should itself be 
non-negative. This important property was not proved in  \cite{BCB73}. A 
central aim of the present paper is to clarify this question, taking into the 
account the available results in the literature. This non-negativity problem 
for the EB system is considered in Section 3. In the case of the Boltzmann 
equation on a given curved spacetime, there are two results known 
\cite{B67,T10} concerning the non-negativity problem. If a solution of the EB 
system is given, in particular a four-dimensional manifold and a metric as a 
part of a solution, a result on non-negativity of solutions of the Boltzmann 
equation on suitable curved spacetimes implies a non-negativity result for the 
EB system. Thus in order to settle this question for the solutions obtained by
Bancel and Choquet-Bruhat in \cite{BCB73} it suffices to examine whether the 
conditions assumed in \cite{B67,T10} hold in the framework of \cite{BCB73}.

In Section 3.1, a result of Tadmon \cite{T10} is considered. Although 
\cite{T10} cites \cite{BCB73} it does not discuss the applicability of the 
results obtained to the solutions of that paper. The conditions in \cite{T10} 
are quite complicated, but it is shown in Section 3.1 that they are satisfied 
in the framework of Bancel and Choquet-Bruhat. In Section 3.2 it is pointed
out that the assumptions assuring non-negativity of solutions of the Boltzmann 
equation in the work of Bichteler \cite{B67} do not follow from those of 
\cite{BCB73}. It is shown, however, that non-negativity of the distribution 
function in the solutions of \cite{BCB73} can be proved using a method simpler 
than that of Tadmon.

Section 4 contains some further general considerations on the EB system. 
The topics discussed include the advantages of a formulation using 
orthonormal frames, the issue of identifying physically relevant
scattering kernels and how the non-negativity result can be made global.

\subsection{Notation}
In this part, we collect the notation which is used in this paper.
Greek indices run from $0$ to $3$, while Latin indices run from $1$ to $3$.
The spatial variable $x$ denotes a four-vector, while the momentum variable
$p$ denotes a three-dimensional vector, i.e.,
\[
x=(x^0,x^1,x^2,x^3),\quad p=(p^1,p^2,p^3).
\]
The metric $g_{\alpha\beta}$ has signature $(-,+,+,+)$, $g$ denotes its 
determinant, and the Minkowski metric is denoted by $\eta_{\alpha\beta}$.
Throughout the paper, the speed of light $c$ and mass of the particles are 
assumed to be unity and hence the momentum $p^\alpha$ lies on a hypersurface 
defined by the equation $g_{\alpha\beta}p^\alpha p^\beta=-1$, which is called the 
mass shell. Due to the mass shell condition, $p^0$ and $p_0$
are represented by $p^i$ as follows.
\begin{equation}\label{2.6}
\begin{aligned}
p^0&=\frac{1}{-g_{00}}\left(g_{0i}p^i+\sqrt{(g_{0i}p^i)^2
-g_{00}(g_{ij}p^ip^j+1)}\right),\\
p_0&=-\sqrt{(g_{0i}p^i)^2-g_{00}(g_{ij}p^ip^j+1)}.
\end{aligned}
\end{equation}
In some places, we use the Euclidean norm in $\bbr^d$, i.e.,
\[
|p|^2=\sum_{i=1}^d(p^i)^2\quad\mbox{{\rm for}}\quad p\in\bbr^d,
\]
and $\cdot$ will denote the usual inner product in $\bbr^d$, i.e.,
\[
p\cdot q=\sum_{i=1}^dp^iq^i\quad\mbox{{\rm for}}\quad p,q\in\bbr^d.
\]

We also collect the function spaces which are used in this paper and which were 
originally introduced in \cite{BCB73}. The spaces are basically Sobolev spaces 
weighted in the momentum variables. Two kinds of weight functions are used,
$(p^0)^{N/2+|\hat{k}|}$ and $e^{p^0}$, where $N$ is a positive integer and 
$|\hat{k}|$ is the number of derivatives which are taken with respect to 
momentum variables. For simplicity, we write $h_k(p)=(p^0)^{N/2+|\hat{k}|}$ and 
$h(p)=e^{p^0}$, which can be understood as the case $N=\infty$.
\begin{enumerate}
\item Let $\omega_0$ be a domain of $\bbr^3$, i.e., $x^0=0$ in $\bbr^4$,
and $\hat{\omega}_0=\omega_0\times\bbr^3$.
\item Let $H_\mu(\omega_0)$ and $H_{\mu,N}(\hat{\omega}_0)$ be the Sobolev 
spaces whose norms are defined as
\[
\|u\|^2_{H_\mu(\omega_0)}=\sum_{|k|\leq\mu}\|D^k_xu\|^2_{L^2(\omega_0)},
\]
\[
\|f\|^2_{H_{\mu,N}(\hat{\omega}_0)}=\sum_{|k|\leq\mu}
\|h_k(p)D^k_{x,p}f\|^2_{L^2(\hat{\omega}_0)},
\]
where $k$ is a multi-index such that $D^k_{x,p}=D^{\bar{k}}_xD^{\hat{k}}_p$ in the 
latter case.
\item Let $\Omega$ be a domain of $\bbr^4$
and $\hat{\Omega}=\Omega\times\bbr^3$.
\item Let $H_\mu(\Omega)$ and $H_{\mu,N}(\hat{\Omega})$ be the Sobolev spaces 
whose norms are defined as
\[
\|u\|^2_{H_\mu(\Omega)}=\sum_{|k|\leq\mu}\|D^k_x u\|^2_{L^2(\Omega)},
\]
\[
\|f\|^2_{H_{\mu,N}(\hat{\Omega})}
=\sum_{|k|\leq\mu}\|h_k(p)D^k_{x,p} f\|^2_{L^2(\hat{\Omega})},
\]
where $k$ is a multi-index such that $D^k_{x,p}=D^{\bar{k}}_xD^{\hat{k}}_p$ in the 
latter case.
\item $g_{\alpha\beta}(0)$, $\partial_0 g_{\alpha\beta}(0)$,
and $f(0)$ will denote the initial data for $g_{\alpha\beta}$ and $f$ 
respectively.
\end{enumerate}

\section{The Boltzmann equation in a curved spacetime}
\setcounter{equation}{0}
In this section, we focus on the Boltzmann part of the EB system.
We assume that a four-dimensional manifold and a metric are given,
and consider the Boltzmann equation in this spacetime. Let $x^\alpha\in\bbr^4$, 
$p^\alpha\in\bbr^4$,
and $p^\alpha$ satisfy $g_{\alpha\beta}p^\alpha p^\beta=-1$.
For a given metric $g_{\alpha\beta}$, the
Boltzmann equation reads 
\[
p^\alpha\frac{\partial f}{\partial x^\alpha}-\Gamma^i_{\alpha\beta}p^\alpha p^\beta
\frac{\partial f}{\partial p^i}=Q(f,f),
\]
\[
Q(f,f)(x,p)=\int_{\bbr^3}\int_{S^2}S(x,p,q,\Omega)
\Big(f(p')f(q')-f(p)f(q)\Big)\,d\Omega\,\frac{|g|^{\frac{1}{2}}}{-q_0}\,dq.
\]
Here $\Gamma^i_{\alpha\beta}$ are the Christoffel symbols defined by
\[
\Gamma^\alpha_{\beta\gamma}=\frac{1}{2}g^{\alpha\delta}
(\partial_\gamma g_{\delta\beta}+\partial_\beta g_{\delta\gamma}-\partial_\delta 
g_{\beta\gamma}),
\]
and $S$ is called the collision cross-section and is a non-negative function,
\begin{equation}\label{2.1}
S=\lambda\varrho\sigma\quad\mbox{{\rm with}}\quad \sigma=\sigma(\varrho,\theta).
\end{equation}
$\Omega$ is a point of $S^2$, thought of as the unit sphere in $\bbr^3$
with standard measure $d\Omega$ and polar angle $\theta$. The quantities 
$\sigma$ and $\theta$ are called the scattering kernel and the scattering 
angle respectively. The scalar quantities $\lambda$ and $\varrho$ are defined by
\begin{equation}\label{2.2}
\lambda=\sqrt{-(p_\alpha+q_\alpha)(p^\alpha+q^\alpha)},\quad
\varrho=\sqrt{(p_\alpha-q_\alpha)(p^\alpha-q^\alpha)}.
\end{equation}
The variables $p'^\alpha$, $q'^\alpha$, $p^\alpha$, and $q^\alpha$ denote the
momenta of two colliding particles, i.e., post-collisional momenta $p'^\alpha$ 
and $q'^\alpha$, and pre-collisional momenta $p^\alpha$ and $q^\alpha$, and they 
satisfy
\[
p'^\alpha+q'^\alpha=p^\alpha+q^\alpha,
\]
which expresses energy-momentum conservation. The collision operator $Q$ can 
be written as $Q=Q_+-Q_-$ where
\begin{align*}
Q_+&=\iint S(x,p,q,\Omega)f(p')f(q')\,d\Omega\frac{|g|^{\frac{1}{2}}}{-q_0}\,dq,\\
Q_-&=f(p)\iint S(x,p,q,\Omega)f(q)\,d\Omega\frac{|g|^{\frac{1}{2}}}{-q_0}\,dq.
\end{align*}
$Q_+$ and $Q_-$ are called the gain term and the loss term respectively.
Note that $Q_+$ and $Q_-$ are non-negative.

In the case $g_{\alpha\beta}=\eta_{\alpha\beta}$, the above Boltzmann equation
reduces to the well-known special relativistic Boltzmann equation.
In this case, the Christoffel symbols vanish, $|g|=1$, and $-q_0=q^0$,
so we obtain
\[
\frac{p^\alpha}{p^0}\frac{\partial f}{\partial x^\alpha}
=\iint v_\phi \sigma(\varrho,\theta)
\Big(f(p')f(q')-f(p)f(q)\Big)\,d\Omega\,dq,
\]
where $v_\phi$ is called the M\o{}ller velocity,
\[
v_\phi=\frac{\lambda\varrho}{p^0q^0}.
\]

The general relativistic Einstein-Boltzmann system can be written in an
alternative way which makes it look more similar to the special relativistic
case and which is more convenient for some purposes. The idea is to 
introduce an orthonormal frame $e^\alpha_\mu$ on spacetime with dual coframe 
$\theta^\mu_\alpha$. Here the indices from the beginning of the Greek alphabet
like $\alpha$ are coordinate indices as before whereas the indices like $\mu$ 
from the middle of the Greek alphabet label the vectors of the frame. In 
formulating the EB system the metric is described by its components in a 
coordinate frame as before but the coordinates $(x^\alpha,p^a)$ used to 
parametrize the mass shell are replaced by coordinates $(x^\alpha,v^u)$ where 
$v^\nu=\theta^\nu_\alpha p^\alpha$. By abuse of notation we will write 
$f(x^\alpha,v^u)$ for the representation of the distribution function in the
new coordinates. An advantage of the new coordinates is that the collision
term no longer contains any explicit dependence on the metric and is 
identical in form to what it is in special relativity. On the other hand
the transport part of the Boltzmann equation  becomes more complicated and the 
Christoffel symbols are replaced by the components of the Levi-Civita 
connection in the chosen frame. Thus we obtain the equation
\begin{equation}
v^\mu e_\mu^\alpha\frac{\partial f}{\partial x^\alpha}
-\eta^{k\sigma}g_{\alpha\beta}e^\alpha_\sigma e^\delta_\lambda
(\partial_\delta e^\beta_\mu+\Gamma^\beta_{\delta\epsilon}e^\epsilon_\mu)
v^\lambda v^\mu\frac{\partial f}{\partial v^k}=Q(f,f).
\end{equation} 

If the EB system is to define a closed system for the quantities
$g_{\alpha\beta}(x^\gamma)$ and $f(x^\gamma,v^u)$ then the orthonormal frame must 
be fixed in terms of the metric and the coordinates by some condition. This 
can be done in such a way that it depends only on the components 
$g_{\alpha\beta}$  and does not contain any direct dependence on the 
coordinates. One way of doing this is to use
the Gram-Schmidt process. This algorithm for producing an orthonormal 
frame from a general frame is best known in the case of positive definite
metrics but it works just as well for pseudo-Riemannian metrics provided
the starting frame is non-degenerate in the sense that the pull-back of
the metric to the subspace spanned by any non-empty subset of the vectors of 
the frame is a non-degenerate quadratic form. This process can be applied to
the coordinate frame $\partial/\partial x^\alpha$ to get the frame 
$e^\alpha_\mu$. Then the components $e^\alpha_\mu$ are algebraic functions
of the components $g_{\alpha\beta}$. In the case of the first two elements
of the basis, for instance, we get
\begin{eqnarray}
&&e_0^\mu=\delta_0^\mu(-g_{00})^{-\frac12}\\
&&e_1^\mu=(-g_{00}\delta_1^\mu+g_{01}\delta_0^\mu)
((g_{00})^2g_{11}-g_{00}(g_{01})^2)^{-\frac12}
\end{eqnarray}  
It then follows that the components $v^\alpha$ are also algebraic functions
of the components $g_{\alpha\beta}$. The coefficients in the transport part of
the Boltzmann equation are algebraic functions of the components 
$g_{\alpha\beta}$ and their first order partial derivatives.

\subsection{The post-collisional momenta}
A parametrization of the post-collisional momenta $p'^\mu$ and $q'^\mu$ in 
general 
relativity has been obtained in \cite{BCB73} using an orthonormal frame
$\{e_\alpha^\mu\}$ associated to a point of the mass shell. Note that this is a 
different type of object from the orthonormal frame introduced above. In 
the frame considered in \cite{BCB73} $p^\mu=e^\mu_\alpha v^\alpha$, 
$q^\mu=e^\mu_\alpha w^\alpha$ and $v^\alpha+w^\alpha=(\lambda,0,0,0)$.
The post-collisional momenta in this frame are represented by
\[
v'^\alpha=\frac{1}{2}(\lambda,\varrho\Omega),\quad
w'^\alpha=\frac{1}{2}(\lambda,-\varrho\Omega),
\]
where $\Omega\in S^2$. Consequently, $p'^\mu$ is given by
\[
p'^\mu=\frac{1}{2}
(e^\mu_0\lambda+e^\mu_1\varrho\cos\theta
+\varrho\sin\theta(e^\mu_2\cos\varphi+e^\mu_3\sin\varphi)),
\]
where $\Omega=(\cos\theta,\sin\theta\cos\varphi,\sin\theta\sin\varphi)$, and
a similar expression is obtained for $q'^\mu$. It is a natural approach to 
choose an orthonormal frame and then use formulae from special relativity, but
in some points it has an inconvenient aspect concerning the orthonormal frame.
In this paper, we consider the Sobolev spaces of order greater than five.
Hence, high order derivatives of some quantities are taken, and we need 
to estimate derivatives of the orthonormal frame chosen. If we follow the 
result of \cite{BCB73}, then derivatives of the post-collisional momenta are 
estimated as follows.
\[
|D^k_{x,p}p'|\leq (p^0)^{-|\hat{k}|}\sum \frac{(p^0q^0)^{|r|+|\hat{s}|
+1/2}}{\lambda^{2|r_1|}\varrho^{2|r_2|}|p\times q|^{|\hat{s}|}}
\left|g^{(|\bar{k}|)}\right|,
\]
where $k$ is a multi-index such that $D^k_{x,p}=D^{\bar{k}}_xD^{\hat{k}}_p$, the 
sum is over all the possible multi-indices $r_1$, $r_2$, $s$ satisfying 
$|r+s|\leq |k|$ and $r_1+r_2=r$, and the term $g^{(|\bar{k}|)}$ will be defined 
later. We can see that singularities appear in the denominator, which are 
certainly due to the choice of the orthonormal frame. This problem might be 
removed by taking another orthonormal frame at the singular points, but it 
would be quite complicated to find such an orthonormal frame that behaves well 
at those singularities. We instead use a different approach to parametrize the 
post-collisional momenta, where the above problem does not arise.

We introduce a new parametrization of post-collisional momenta.
Suppose that $p^\alpha$ and $q^\alpha$ are given, and consider the following 
four-vectors.
\begin{equation}\label{2.3}
n^\alpha=p^\alpha+q^\alpha\quad\mbox{{\rm and}}\quad
t^\alpha=(n_i\omega^i,-n_0\omega)\quad\mbox{{\rm for}}\quad \omega\in S^2,
\end{equation}
Note that $t^\alpha$ is orthogonal to $n^\alpha$.
Then $p'^\alpha$ and $q'^\alpha$ can be parametrized by
\begin{equation}\label{2.4}
\begin{aligned}
&p'^\alpha=p^\alpha-\frac{t_\beta(p^\beta-q^\beta)}{t_\beta t^\beta}t^\alpha
=p^\alpha+2\frac{t_\beta q^\beta}{t_\beta t^\beta}t^\alpha,\\
&q'^\alpha=q^\alpha+\frac{t_\beta(p^\beta-q^\beta)}{t_\beta t^\beta}t^\alpha
=q^\alpha-2\frac{t_\beta q^\beta}{t_\beta t^\beta}t^\alpha,
\end{aligned}
\end{equation}
where we used $t_\alpha n^\alpha=0$. It can be easily shown that they satisfy
\[
p'^\alpha+q'^\alpha=p^\alpha+q^\alpha\quad\mbox{{\rm and}}\quad
p'_\alpha p'^\alpha=q'_\alpha q'^\alpha=-1.
\]
We can use another parametrization for $p'^\alpha$ and $q'^\alpha$,
\begin{equation}\label{2.5}
p'^\alpha=\frac{p^\alpha+q^\alpha}{2}
+\frac{\varrho}{2}\frac{t^\alpha}{\sqrt{t_\beta t^\beta}},\quad
q'^\alpha=\frac{p^\alpha+q^\alpha}{2}
-\frac{\varrho}{2}\frac{t^\alpha}{\sqrt{t_\beta t^\beta}},
\end{equation}
where $\varrho$ and $t^\alpha$ are defined by \eqref{2.2} and \eqref{2.3} 
respectively.
We can see that the parametrizations \eqref{2.4} and \eqref{2.5} are very
similar to those of the nonrelativistic case. In the nonrelativistic case, 
post-collisional momenta $p'\in\bbr^3$ and $q'\in\bbr^3$ are represented in 
two ways.
\[
p'=p-((p-q)\cdot\omega)\omega,\quad
q'=q+((p-q)\cdot\omega)\omega\quad\mbox{{\rm for}}\quad \omega\in S^2
\]
or
\[
p'=\frac{p+q}{2}+\frac{|p-q|}{2}\omega,\quad
q'=\frac{p+q}{2}-\frac{|p-q|}{2}\omega
\quad\mbox{{\rm for}}\quad \omega\in S^2.
\]
Hence, \eqref{2.4} and \eqref{2.5} can be thought of as natural
generalizations from the nonrelativistic case.

We remark that a well-known parametrization of $p'^\alpha$ and $q'^\alpha$ in 
the special relativistic case by Glassey and Strauss \cite{GS91} is exactly the 
same as \eqref{2.4}. In the special relativistic case, 
$g_{\alpha\beta}=\eta_{\alpha\beta}$, we have
\[
t_\alpha t^\alpha=-(n\cdot\omega)^2+(n^0)^2=(p^0+q^0)^2-((p+q)\cdot\omega)^2,
\]
and by direct calculations,
\begin{align*}
t_\alpha q^\alpha &= -q^0n\cdot\omega +n^0\omega\cdot q
=-q^0(p\cdot\omega)-q^0(q\cdot\omega)+p^0(\omega\cdot q)+q^0(\omega\cdot q)\\
&=-p^0q^0(p/p^0\cdot\omega)+p^0q^0(q/q^0\cdot\omega)
=-p^0q^0((\hat{p}-\hat{q})\cdot\omega).
\end{align*}
Defining $\hat{p}=p/p^0$ and $\hat{q}=q/q^0$ we obtain the following 
expression.
\[
p'=p-\frac{2p^0q^0((\hat{p}-\hat{q})\cdot\omega)(p^0+q^0)}{(p^0+q^0)^2
-((p+q)\cdot\omega)^2}\omega.
\]
This is the parametrization given in \cite{GS91}.

We close this subsection with a simple lemma, which shows that
$p^0$ and $-p_0$ are equivalent in a curved spacetime when its metric is
close to the Minkowski metric.
\begin{lemma}
Consider a momentum $p^\alpha$ satisfying $g_{\alpha\beta}p^\alpha p^\beta=-1$.
Suppose that there exists a small $\varepsilon$ such that the metric 
$g_{\alpha\beta}$
satisfies
\[
|g_{\alpha\beta}-\eta_{\alpha\beta}|\leq \varepsilon\quad\mbox{{\rm and}}\quad
(1-\varepsilon)\sum_{i=1}^3(X^i)^2\leq g_{ij}X^i X^j\leq
(1+\varepsilon)\sum_{i=1}^3(X^i)^2
\]
for any three dimensional vector $X$. Then we have
\[
|p|\leq (1+\varepsilon_1)\min\{p^0,-p_0\}\quad\mbox{{\rm and}}\quad
(1-\varepsilon_1)p^0\leq -p_0\leq (1+\varepsilon_1)p^0
\]
for some small $\varepsilon_1>0$.
\end{lemma}
\begin{proof}
By solving the equation $g_{\alpha\beta}p^\alpha p^\beta=-1$,
$p^0$ and $p_0$ can be expressed as functions of $g_{\alpha\beta}$ and $p^i$
as in (\ref{2.6}). Hence, by direct calculations,
\begin{align*}
|p|^2\leq (1+\varepsilon)g_{ij}p^ip^j
\leq (1+\varepsilon)((g_{0i}p^i)^2-g_{00}(g_{ij}p^ip^j+1))
=(1+\varepsilon)(p_0)^2,
\end{align*}
\[
-p_0=-g_{00}p^0-g_{0i}p^i
\leq (1+\varepsilon)p^0+\varepsilon|p|,
\]
which imply that $|p|\leq Cp_0$ and $|p|\leq Cp^0$ for some constant 
$C\approx1$.
From the identity $p_0=g_{00}p^0+g_{0i}p^i$ and the conditions on $g_{\alpha\beta}$,
we get the desired equivalence between $p^0$ and $-p_0$.
\end{proof}

\subsection{Derivatives of the post-collisional momenta}
In this paper, we use \eqref{2.4} for $p'^\alpha$ and $q'^\alpha$.
Since the parametrization \eqref{2.4} depends on $x$, $p$, and $q$,
derivatives of the post-collisional momenta are not trivial to compute,
but we can see that derivatives of $p'^\alpha$ and $q'^\alpha$
depend only on $t^\alpha$. To calculate derivatives of $p'^\alpha$ and $q'^\alpha$,
we use the notation $g^{(m)}$ of Bancel and Choquet-Bruhat from \cite{BCB73}.

A function $h$ is said to be a $g^{(m)}$ function, if $h$ is any linear 
combination of products,
$\prod_j D^{k_j} g_{\mu\nu}$,
with $\sum_j|k_j|\leq m$ and coefficients
in the algebra of bounded functions on $\Omega\times\bbr^3\times\bbr^3$
generated by $p^\alpha/p^0$, $p^\alpha/p_0$, $q^\alpha/q^0$, and $q^\alpha/q_0$.

\begin{lemma}
Let $t^\alpha$ be a four-vector defined by \eqref{2.3} for some $\omega\in S^2$.
Then, for $D^k=D^{\bar{k}}_x D^{\hat{k}}_p$ with multi-index $k=\bar{k}+\hat{k}$, 
we have
\[
D^k\left[2\frac{t_\beta q^\beta}{t_\beta t^\beta}t^\alpha\right]
=(p^0)^{-|\hat{k}|}\sum_{i=0}^{2+2|k|}
\frac{(p^0)^i(q^0)^{3+2|k|-i}}{(t_\alpha t^\alpha)^{1+|k|}}
g_i^{(|\bar{k}|)},
\]
where $g_i^{(|\bar{k}|)}$ are some $g^{(|\bar{k}|)}$ functions.
\end{lemma}
\begin{proof}
We first note that $p$-derivatives of $p^\alpha$ and $p_\alpha$ are $g^{(0)}$ 
functions as follows,
\[
\frac{\partial p^0}{\partial p^k}=-\frac{p_k}{p_0}
=-\frac{p^\alpha}{p_0}g_{\alpha k},\qquad
\frac{\partial p^i}{\partial p^k}=\delta^i_k,
\]
\[
\frac{\partial p_\alpha}{\partial p^k}=\frac{\partial p^0}{\partial p^k}g_{0\alpha}
+\frac{\partial p^i}{\partial p^k}g_{i\alpha}
=-\frac{p^\beta}{p_0}g_{\beta k}g_{0\alpha}+g_{k\alpha},
\]
while their $x^\alpha$-derivatives are $p^0g^{(1)}$ functions,
\[
\frac{\partial p^0}{\partial x^\alpha}
=-\frac{p^\beta p^\gamma}{2p_0}\partial_\alpha g_{\beta\gamma},\qquad
\frac{\partial p_\beta}{\partial x^\alpha}
=\frac{\partial p^0}{\partial x^\alpha}g_{0\beta}
+p^\gamma\partial_\alpha g_{\beta\gamma}
=-\frac{p^\gamma p^\delta}{2p_0}\partial_\alpha g_{\gamma\delta}g_{0\beta}
+p^\gamma\partial_\alpha g_{\beta\gamma}.
\]
Hence, we can deduce that
\[
\frac{\partial g^{(m)}}{\partial p^k}=\frac{g^{(m)}}{p^0}\quad\mbox{{\rm and}}\quad
\frac{\partial g^{(m)}}{\partial x^\alpha}=g^{(m+1)}.
\]
To prove the lemma, we use induction on $|k|$.
Since $n^\alpha=p^\alpha+q^\alpha$, the four-vector $t^\alpha$ is of the form 
\[
t^\alpha=(n_i\omega^i,-n_0\omega)=n^\beta(g_{i\beta}\omega^i,-g_{0\beta}\omega)
=p^0g^{(0)}_0+q^0g^{(0)}_1
\]
for some $g^{(0)}_0$ and $g^{(0)}_1$, which are $g^{(0)}$ functions.
Hence, we obtain for $|k|=0$,
\[
2\frac{t_\beta q^\beta}{t_\beta t^\beta}t^\alpha
=\sum_{i=0}^{2}\frac{(p^0)^i(q^0)^{3-i}}{t_\alpha t^\alpha}g^{(0)}_i.
\]
We use $t_\alpha t^\alpha=\sum_{i=0}^{2}(p^0)^i(q^0)^{2-i}g^{(0)}_i$
to show that for $|k|+1$.
\begin{align*}
D_{x^\alpha}D^k\left[2\frac{t_\beta q^\beta}{t_\beta t^\beta}t^\gamma\right]
&=D_{x^\alpha}\left[(p^0)^{-|\hat{k}|}\sum_{i=0}^{2+2|k|}
\frac{(p^0)^i(q^0)^{3+2|k|-i}}{(t_\alpha t^\alpha)^{1+|k|}}
g_i^{(|\bar{k}|)}
\right]\\
&=(p^0)^{-|\hat{k}|}\sum_{i=0}^{2+2|k|}
\frac{(p^0)^i(q^0)^{3+2|k|-i}}{(t_\beta t^\beta)^{1+|k|}}
g_i^{(|\bar{k}|+1)}\\
&\quad+(p^0)^{-|\hat{k}|}\sum_{i=0}^{2+2|k|}
\frac{(p^0)^i(q^0)^{3+2|k|-i}}{(t_\beta t^\beta)^{2+|k|}}
D_{x^\alpha}\left[t_\beta t^\beta\right]g_i^{(|\bar{k}|)}.
\end{align*}
Since $D_{x^\alpha}\left[t_\beta t^\beta\right]
=\sum_{i=0}^{2}(p^0)^i(q^0)^{2-i}g^{(1)}_i$,
we obtain
\[
D_{x^\alpha}D^k\left[2\frac{t_\beta q^\beta}{t_\beta t^\beta}t^\gamma\right]
=(p^0)^{-|\hat{k}|}\sum_{i=0}^{4+2|k|}
\frac{(p^0)^i(q^0)^{5+2|k|-i}}{(t_\beta t^\beta)^{2+|k|}}
g_i^{(|\bar{k}|+1)}.
\]
Similarly,
\begin{align*}
D_{p^j}D^k\left[2\frac{t_\beta q^\beta}{t_\beta t^\beta}t^\gamma\right]
&=D_{p^j}\left[(p^0)^{-|\hat{k}|}\sum_{i=0}^{2+2|k|}
\frac{(p^0)^i(q^0)^{3+2|k|-i}}{(t_\alpha t^\alpha)^{1+|k|}}
g_i^{(|\bar{k}|)}
\right]\\
&=(p^0)^{-|\hat{k}|-1}\sum_{i=0}^{2+2|k|}
\frac{(p^0)^i(q^0)^{3+2|k|-i}}{(t_\beta t^\beta)^{1+|k|}}
g_i^{(|\bar{k}|)}\\
&\quad+(p^0)^{-|\hat{k}|}\sum_{i=0}^{2+2|k|}
\frac{(p^0)^i(q^0)^{3+2|k|-i}}{(t_\beta t^\beta)^{2+|k|}}
D_{p^j}\left[t_\beta t^\beta\right]g_i^{(|\bar{k}|)}.
\end{align*}
Since $D_{p^j}\left[t_\beta t^\beta\right]
=(p^0)^{-1}\sum_{i=0}^{2}(p^0)^i(q^0)^{2-i}g^{(0)}_i$,
we obtain
\[
D_{x^\alpha}D^k\left[2\frac{t_\beta q^\beta}{t_\beta t^\beta}t^\gamma\right]
=(p^0)^{-|\hat{k}|-1}\sum_{i=0}^{4+2|k|}
\frac{(p^0)^i(q^0)^{5+2|k|-i}}{(t_\beta t^\beta)^{2+|k|}}
g_i^{(|\bar{k}|)},
\]
and this completes the proof of the lemma.
\end{proof}

To estimate derivatives of the post-collisional momenta, we need to obtain a 
lower bound for the quantity $t_\alpha t^\alpha$, which appears
in the denominator in the expressions coming from Lemma 2.2. Note that 
$t_\alpha t^\alpha$ is non-negative, because $t_\alpha n^\alpha=0$ and 
$n^\alpha=p^\alpha+q^\alpha$ is a timelike vector.

\begin{lemma}
Let $t^\alpha$ be a four-vector defined by \eqref{2.3} for some $\omega\in S^2$.
Suppose that there exists a small $\varepsilon$ such that the metric 
$g_{\alpha\beta}$ satisfies
\[
|g_{\alpha\beta}-\eta_{\alpha\beta}|\leq \varepsilon\quad\mbox{{\rm and}}\quad
(1-\varepsilon)\sum_{i=1}^3(X^i)^2\leq g_{ij}X^i X^j\leq
(1+\varepsilon)\sum_{i=1}^3(X^i)^2
\]
for any three dimensional vector $X$, then we have the following lower bound.
\[
t_\alpha t^\alpha \geq 2(g_{0i}\omega^i)^2-g_{00}(g_{ij}\omega^i\omega^j)
\left(2+\frac{(p_0)^2+(q_0)^2}{3p_0q_0}\right).
\]
\end{lemma}
\begin{proof}
The proof is a direct calculation. Since $t^\alpha=(n_i\omega^i,-n_0\omega)$, 
we have
\begin{align}
t_\alpha t^\alpha &=(n_0)^2(g_{ij}\omega^i\omega^j)
-2(n_i\omega^i)n_0(g_{0i}\omega^i)
+g_{00}(n_i\omega^i)^2\cr
&=(n_0)^2(g_{ij}\omega^i\omega^j)
+g_{00}((g_{ij}n^i\omega^j)^2+2(g_{ij}n^i\omega^j)n^0(g_{0i}\omega^i)
+(n^0)^2(g_{0i}\omega^i)^2)\cr
&\quad -2((g_{ij}n^i\omega^j)n_0(g_{0i}\omega^i)+n^0n_0(g_{0i}\omega^i)^2),
\label{2.7}
\end{align}
where we used $n_i\omega^i=g_{ij}n^i\omega^j+n^0(g_{0i}\omega^i)$.
There are six terms on the right hand side of \eqref{2.7}. The first two of
these terms can be estimated as follows using the Cauchy-Schwarz inequality.
\begin{align*}
(n_0)^2(g_{ij}\omega^i\omega^j)+g_{00}(g_{ij}n^i\omega^j)^2
&\geq (n_0)^2(g_{ij}\omega^i\omega^j)+g_{00}(g_{ij}n^in^j)(g_{ij}\omega^i\omega^j),
\end{align*}
We continue the estimation as follows
\begin{align}
&\hspace{-0.5cm}(n_0)^2(g_{ij}\omega^i\omega^j)+g_{00}(g_{ij}n^i\omega^j)^2\cr
&\geq (g_{ij}\omega^i\omega^j)
((p_0)^2+(q_0)^2+2p_0q_0+g_{00}(g_{ij}p^ip^j)+g_{00}(g_{ij}q^iq^j)
+2g_{00}(g_{ij}p^iq^j))\cr
&=(g_{ij}\omega^i\omega^j)
((g_{0i}p^i)^2+(g_{0i}q^i)^2-2g_{00}+2p_0q_0+2g_{00}(g_{ij}p^iq^j)),\label{2.8}
\end{align}
using
\[
p_0=-\sqrt{(g_{0i}p^i)^2-g_{00}((g_{ij}p^ip^j)+1)}
\]
and the analogous formula for $q_0$. The fourth and the sixth terms
on the right hand side of \eqref{2.7} are calculated as follows.
\begin{align*}
&\hspace{-0.5cm}(g_{0i}\omega^i)^2(g_{00}(n^0)^2-2n^0n_0)\\
&=(g_{0i}\omega^i)^2(g_{00}(p^0)^2+g_{00}(q^0)^2+2g_{00}p^0q^0
-2p^0p_0-2p^0q_0-2q^0p_0-2q^0q_0)\\
&=(g_{0i}\omega^i)^2(-g_{00}(p^0)^2-g_{00}(q^0)^2-2g_{00}p^0q^0
-2p^0(g_{0i}n^i)-2q^0(g_{0i}n^i)),
\end{align*}
where we used $p_0=g_{00}p^0+(g_{0i}p^i)$ and the analogous formula for $q_0$.
To proceed with the calculation, we use the formula 
\[
p^0=\frac{1}{-g_{00}}\left(g_{0i}p^i+\sqrt{(g_{0i}p^i)^2
-g_{00}(g_{ij}p^ip^j+1)}\right)
\]
and the corresponding formula for $q^0$. After a long calculation we obtain
\begin{align}
&\hspace{-0.5cm}(g_{0i}\omega^i)^2(g_{00}(n^0)^2-2n^0n_0)\cr
&=(g_{0i}\omega^i)^2((g_{ij}p^ip^j)+(g_{ij}q^iq^j)+2)\cr
&\quad+2(g_{0i}\omega^i)^2\frac{1}{-g_{00}}
\sqrt{(g_{0i}p^i)^2-g_{00}((g_{ij}p^ip^j)+1)}\sqrt{(g_{0i}q^i)^2
-g_{00}((g_{ij}q^iq^j)+1)}\cr
&\quad-2(g_{0i}\omega^i)^2\frac{1}{-g_{00}}(g_{0i}p^i)(g_{0i}q^i).\label{2.9}
\end{align}
The third and the fifth terms on the right hand side of \eqref{2.7} are 
estimated as follows.
\begin{align}
(g_{0i}\omega^i)(g_{ij}n^i\omega^j)(2g_{00}n^0-2n_0)
&=-2(g_{0i}\omega^i)(g_{ij}n^i\omega^j)(g_{0i}n^i)\cr
&\geq -2|g_{0i}\omega^i|\sqrt{g_{ij}n^in^j}
\sqrt{g_{ij}\omega^i\omega^j}|g_{0i}n^i|\cr
&\geq -(g_{0i}\omega^i)^2(g_{ij}n^in^j)-(g_{ij}\omega^i\omega^j)(g_{0i}n^i)^2,
\label{2.10}
\end{align}
where we used the Cauchy-Schwarz inequality. Hence all terms on the right
hand side of  \eqref{2.7} have been estimated by \eqref{2.8}--\eqref{2.10}.
Moreover, since $n^\alpha=p^\alpha+q^\alpha$, the first two terms on the right 
hand side of \eqref{2.8} are partially cancelled by the second term on the 
right hand side of \eqref{2.10}, and similarly the first two terms on the 
right hand side of \eqref{2.9} are partially cancelled by the first
term on the right hand side of \eqref{2.10} as follows.
\begin{align}
(g_{ij}\omega^i\omega^j)((g_{0i}p^i)^2+(g_{0i}q^i)^2-(g_{0i}n^i)^2)
=-2(g_{ij}\omega^i\omega^j)(g_{0i}p^i)(g_{0i}q^i),\label{2.11}
\end{align}
and
\begin{align}
(g_{0i}\omega^i)^2((g_{ij}p^ip^j)+(g_{ij}q^iq^j)-(g_{ij}n^in^j))
&=-2(g_{0i}\omega^i)^2(g_{ij}p^iq^j)\cr
&\geq -2(g_{0i}\omega^i)^2
\sqrt{g_{ij}p^ip^j}\sqrt{g_{ij}q^iq^j}.\label{2.12}
\end{align}
As a consequence, \eqref{2.7} is estimated by \eqref{2.11}--\eqref{2.12} and
the other terms of \eqref{2.8}--\eqref{2.9}, i.e.,
\begin{align*}
t_\alpha t^\alpha
&\geq (g_{ij}\omega^i\omega^j)(-2g_{00}+2p_0q_0
+2g_{00}(g_{ij}p^iq^j)-2(g_{0i}p^i)(g_{0i}q^i))\cr
&\quad+(g_{0i}\omega^i)^2\left(-2\sqrt{g_{ij}p^ip^j}\sqrt{g_{ij}q^iq^j}+2\right)\cr
&\quad+2(g_{0i}\omega^i)^2\frac{1}{-g_{00}}
\sqrt{(g_{0i}p^i)^2-g_{00}((g_{ij}p^ip^j)+1)}\sqrt{(g_{0i}q^i)^2
-g_{00}((g_{ij}q^iq^j)+1)}\cr
&\quad-2(g_{0i}\omega^i)^2\frac{1}{-g_{00}}(g_{0i}p^i)(g_{0i}q^i),
\end{align*}
and we rewrite it as follows.
\begin{align}
t_\alpha t^\alpha
&\geq-2(g_{ij}\omega^i\omega^j)g_{00}+2(g_{0i}\omega^i)^2\cr
&\quad+2(g_{ij}\omega^i\omega^j)(p_0q_0+g_{00}(g_{ij}p^iq^j)
-(g_{0i}p^i)(g_{0i}q^i))\cr
&\quad+2(g_{0i}\omega^i)^2\frac{1}{-g_{00}}
\sqrt{(g_{0i}p^i)^2-g_{00}((g_{ij}p^ip^j)+1)}
\sqrt{(g_{0i}q^i)^2-g_{00}((g_{ij}q^iq^j)+1)}\cr
&\quad-2(g_{0i}\omega^i)^2\frac{1}{-g_{00}}\left(-g_{00}
\sqrt{g_{ij}p^ip^j}\sqrt{g_{ij}q^iq^j}
+(g_{0i}p^i)(g_{0i}q^i)\right).\label{2.13}
\end{align}
Note that the first line of \eqref{2.13} is strictly positive, and
the sum of the third and the fourth lines
is non-negative because of the following calculation.
\begin{align*}
&\hspace{-0.5cm}\left(-g_{00}\sqrt{g_{ij}p^ip^j}\sqrt{g_{ij}q^iq^j}
+(g_{0i}p^i)(g_{0i}q^i)\right)^2\\
&=(g_{00})^2(g_{ij}p^ip^j)(g_{ij}q^iq^j)
+(g_{0i}p^i)^2(g_{0i}q^i)^2-2g_{00}(g_{0i}p^i)(g_{0i}q^i)
\sqrt{g_{ij}p^ip^j}\sqrt{g_{ij}q^iq^j}\\
&\leq (g_{00})^2(g_{ij}p^ip^j)(g_{ij}q^iq^j)
+(g_{0i}p^i)^2(g_{0i}q^i)^2
-g_{00}((g_{0i}p^i)^2(g_{ij}q^iq^j)
+(g_{0i}q^i)^2(g_{ij}p^ip^j)).
\end{align*}
We now have only the second line of \eqref{2.13}, which is explicitly 
calculated as follows.
Note that the coefficient $g_{ij}\omega^i\omega^j$ is strictly positive, and we 
have
\begin{align}
&\hspace{-0.5cm}p_0q_0+g_{00}(g_{ij}p^iq^j)-(g_{0i}p^i)(g_{0i}q^i)\cr
&\geq p_0q_0-\left(-g_{00}\sqrt{g_{ij}p^ip^j}\sqrt{g_{ij}q^iq^j}
+|g_{0i}p^i||g_{0i}q^i|\right)\cr
&=\frac{(p_0q_0)^2-\left(-g_{00}\sqrt{g_{ij}p^ip^j}
\sqrt{g_{ij}q^iq^j}+|g_{0i}p^i||g_{0i}q^i|\right)^2
}{p_0q_0+\left(-g_{00}\sqrt{g_{ij}p^ip^j}
\sqrt{g_{ij}q^iq^j}+|g_{0i}p^i||g_{0i}q^i|\right)}\cr
&\geq\frac{1}{3p_0q_0}\left((p_0q_0)^2-\left(-g_{00}\sqrt{g_{ij}p^ip^j}
\sqrt{g_{ij}q^iq^j}+|g_{0i}p^i||g_{0i}q^i|\right)^2\right),\label{2.14}
\end{align}
where we used Lemma 2.1.
Note that
\begin{align*}
(p_0q_0)^2
&=(g_{0i}p^i)^2(g_{0i}q^i)^2\cr
&\quad -g_{00}(g_{0i}p^i)^2(g_{ij}q^iq^j)-g_{00}(g_{0i}p^i)^2
-g_{00}(g_{0i}q^i)^2(g_{ij}p^ip^j)-g_{00}(g_{0i}q^i)^2\cr
&\quad+(g_{00})^2((g_{ij}p^ip^j)(g_{ij}q^iq^j)+g_{ij}p^ip^j+g_{ij}q^iq^j+1),
\end{align*}
and
\begin{align*}
&\hspace{-0.5cm}\left(-g_{00}\sqrt{g_{ij}p^ip^j}
\sqrt{g_{ij}q^iq^j}+|g_{0i}p^i||g_{0i}q^i|\right)^2\\
&=(g_{00})^2(g_{ij}p^ip^j)(g_{ij}q^iq^j)
-2g_{00}|g_{0i}p^i||g_{0i}q^i|
\sqrt{g_{ij}p^ip^j}\sqrt{g_{ij}q^iq^j}
+(g_{0i}p^i)^2(g_{0i}q^i)^2\\
&\leq (g_{00})^2(g_{ij}p^ip^j)(g_{ij}q^iq^j)
-g_{00}(g_{0i}p^i)^2(g_{ij}q^iq^j)
-g_{00}(g_{0i}q^i)^2(g_{ij}p^ip^j)
+(g_{0i}p^i)^2(g_{0i}q^i)^2.
\end{align*}
Then, \eqref{2.14} is estimated as follows.
\begin{align}
&\hspace{-0.5cm}p_0q_0+g_{00}(g_{ij}p^iq^j)-(g_{0i}p^i)(g_{0i}q^i)\cr
&\geq\frac{1}{3p_0q_0}(-g_{00}(g_{0i}p^i)^2-g_{00}(g_{0i}q^i)^2
+(g_{00})^2(g_{ij}p^ip^j+g_{ij}q^iq^j+1))\cr
&\geq\frac{-g_{00}}{6p_0q_0}((p_0)^2+(q_0)^2),\label{2.15}
\end{align}
where we used the explicit formula \eqref{2.6} for $p_0$ and $q_0$.
We apply the estimate \eqref{2.15} to \eqref{2.13},
and obtain the desired lower bound.
\end{proof}

\begin{remark}
In the special relativistic case, we have a more refined result,
\[
t_\alpha t^\alpha\geq 2+|p\times\omega|^2+|q\times\omega|^2
+\frac{1+|p|^2+|q|^2}{\sqrt{1+|p|^2}\sqrt{1+|q|^2}},
\]
which was obtained in \cite{GS91}. This inequality implies that 
$t_\alpha t^\alpha\geq p^0/q^0$
in the special relativistic case, and Lemma 2.3 shows that a similar 
inequality holds in a curved
spacetime, i.e.,
\begin{equation}
t_\alpha t^\alpha \geq Cp_0/q_0\quad\mbox{{\rm or}}\quad
t_\alpha t^\alpha \geq Cp^0/q^0\label{2.16}
\end{equation}
for some positive constant $C$.
\end{remark}

We apply the inequality \eqref{2.16} to the result of Lemma 2.2,
and obtain the following estimate. Under the same assumptions as in Lemma 2.2, 
we have
\begin{align*}
\left|D^k\left[2\frac{t_\beta q^\beta}{t_\beta t^\beta}t^\alpha\right]\right|
&\leq (p^0)^{-1-2|k|}\sum_{i=0}^{2+2|k|}(p^0)^{i+|\bar{k}|}(q^0)^{4+3|k|-i}
\left|g_i^{(|\bar{k}|)}\right|\\
&\leq C (p^0)^{1+|\bar{k}|}(q^0)^{4+3|k|}
\left|g^{(|\bar{k}|)}\right|
\end{align*}
for some $g^{(|\bar{k}|)}$ and a positive constant $C$.
Since $p'^\alpha$ and $q'^\alpha$ are parametrized by \eqref{2.4} and it can be 
easily shown that
$D^k p^\alpha =(p^0)^{1-|\hat{k}|}g^{(|\bar{k}|)}$,
we obtain the following result on derivatives of the post-collisional momenta.
\begin{lemma}
Let $t^\alpha$ be a four-vector defined by \eqref{2.3} for some $\omega\in S^2$,
and $p'^\alpha$ and $q'^\alpha$ be two post-collisional momenta parametrized by 
\eqref{2.4}.
Suppose that there exists a small $\varepsilon$ such that the metric 
$g_{\alpha\beta}$ satisfies
\[
|g_{\alpha\beta}-\eta_{\alpha\beta}|\leq \varepsilon\quad\mbox{{\rm and}}\quad
(1-\varepsilon)\sum_{i=1}^3(X^i)^2\leq g_{ij}X^i X^j\leq
(1+\varepsilon)\sum_{i=1}^3(X^i)^2
\]
for any three dimensional vector $X$.
Then, derivatives of $p'^\alpha$ and $q'^\alpha$ can be estimated as follows.
\[
|D^kp'^\alpha|+|D^kq'^\alpha|\leq C (p^0)^{1+|\bar{k}|}(q^0)^{4+3|k|}
\left|g^{(|\bar{k}|)}\right|
\]
for some $g^{(|\bar{k}|)}$ and
$D^k=D^{\bar{k}}_xD^{\hat{k}}_p$ with multi-index $k=\bar{k}+\hat{k}$.
\end{lemma}

\subsection{The $\mu-N$ regularity of the collision operator}
In \cite{BCB73}, the authors introduced a regularity property of $Q$
called $\mu-N$ regularity, and showed that the Cauchy problem for the EB 
system is well-posed when $Q$ satisfies it. The collision operator $Q$ is said 
to satisfy $\mu-N$ regularity if there exists a constant $C$ satisfying
for each $\hat{\omega}_t$,
\[
\left\|\frac{1}{p^0}Q(f,f)(t)\right\|_{H_{\mu,N}(\hat{\omega}_t)}
\leq C\|f(t)\|^2_{H_{\mu,N}(\hat{\omega}_t)}.
\]
To obtain $\mu-N$ regularity, the collision cross-section $S$ should 
satisfy suitable conditions, and by following the calculations of 
\cite{BCB73}, we can see that the following inequality also holds under the 
same conditions on $S$.
\begin{equation}\label{2.17}
\left\|\frac{1}{p^0}Q(f,g)(t)\right\|_{H_{\mu,N}(\hat{\omega}_t)}
\leq C\|f(t)\|_{H_{\mu,N}(\hat{\omega}_t)}
\|g(t)\|_{H_{\mu,N}(\hat{\omega}_t)},
\end{equation}
where $Q(f,g)$ is defined by
\[
Q(f,g)=\int_{\bbr^3}\int_{S^2}S(x,p,q,\Omega)\Big(
f(p')g(q')-f(p)g(q)\Big)\,d\Omega\,\frac{|g|^{\frac{1}{2}}}{-q_0}\,dq.
\]

In the previous section, we obtained a new estimate of $D^kp'^\alpha$ and 
$D^kq'^\alpha$ in Lemma 2.4. This leads to a corresponding new set
of conditions on $S$ which imply the $\mu-N$ regularity of $Q$.
Note that high order derivatives $D^k_{x,p}\big[f(x,p')\big]$ are linear 
combinations of the following quantities.
\begin{equation}
(D^if)(x,p')(D^{j_1}p')\cdots(D^{j_i}p')\quad\mbox{{\rm with}}\quad
j_1+\cdots+j_i=|k|,\label{2.18}
\end{equation}
where $D^i$ denotes some differential operator $D^r_{x,p}$ satisfying $|r|=i$.
Lemma 2.4 gives the estimate
\[
\left|D^k_{x,p}\big[f(x,p')\big]\right|
\leq C\sum|(D^if)(x,p')|(p^0)^{i+|\bar{k}|}(q^0)^{4i+3|\bar{k}|}
\left|g^{(|\bar{k}|)}\right|,
\]
where the sum is over all the possible $i$ satisfying \eqref{2.18}.
Consider the following quantity which arises from estimation of the gain term 
$Q_+$.
Let the multi-indices $k$, $r$, and $s$ satisfy $k+r+s=l$ with $|l|\leq\mu$.
\begin{align*}
&\hspace{-0.5cm}\left(\frac{1}{p^0}
\iint (D^k_{x,p}S) D^r_{x,p}\big[f(p')\big] D^s_{x,p}\big[f(q')\big]
\,d\Omega\,\frac{|g|^{\frac{1}{2}}}{-q_0}\,dq\right)^2\\
&\leq C\sum
\left(\iint |D^k_{x,p}S|
|D^if(p')|(p^0)^{i+|\bar{r}|}(q^0)^{4i+3|\bar{r}|}\left|g^{(|\bar{r}|)}\right|\right.\\
&\hspace{4.5cm}\left.\times
|D^jf(q')|(p^0)^{j+|\bar{s}|}(q^0)^{4j+3|\bar{s}|}\left|g^{(|\bar{s}|)}\right|
\,d\Omega\,\frac{1}{p_0q_0}\,dq\right)^2\\
&\leq C\sum\iint h_l^{-2}(p) h_i^2(p')|D^if(p')|^2\left|g^{(|\bar{r}|)}\right|^2
h_j^2(q')|D^jf(q')|^2\left|g^{(|\bar{s}|)}\right|^2
d\Omega\,\frac{p'_0q'_0}{p_0q_0}\,dq\\
&\hspace{1.5cm}\times\iint |D^k_{x,p}S|^2
h_l^2(p)h_i^{-2}(p')h_j^{-2}(q')(p^0)^{2(i+j)+2|\overline{r+s}|}
(q^0)^{8(i+j)+6|\overline{r+s}|}\,d\Omega\,
\frac{1}{p_0'q_0'p_0q_0}\,dq,
\end{align*}
and if the second integral is bounded, then we multiply the first integral
by the weight function $h_l^2(p)$ and integrate it over $x$ and $p$ with 
$\frac{p_0'q_0'}{p_0q_0}\,dp\,dq=dp'\,dq'$,
to obtain the inequality \eqref{2.17} when $\mu\geq 5$. Hence, the 
boundedness of the second integral should be the condition on $S$ for the 
$\mu-N$ regularity of $Q$. The second integral is estimated as follows. In the 
case where $N$ is finite,
\begin{align}
&\hspace{-0.5cm}\iint |D^k_{x,p}S|^2
h_l^2(p)h_i^{-2}(p')h_j^{-2}(q')(p^0)^{2(i+j)+2|\overline{r+s}|}
(q^0)^{8(i+j)+6|\overline{r+s}|}\,d\Omega\,
\frac{1}{p_0'q_0'p_0q_0}\,dq\cr
&\leq C\iint |D^k_{x,p}S|^2
\frac{(p^0)^{N+2|\hat{l}|}}{(p'^0)^{N+2\hat{i}}(q'^0)^{N+2\hat{j}}}
(p^0)^{2|r+s|+2|\overline{r+s}|-2}(q^0)^{14|r+s|-1}
\,d\Omega\,dq\cr
&\leq C\iint |D^k_{x,p}S|^2(p^0)^{2|\hat{l}|}
(p^0)^{2|r+s|+2|\overline{r+s}|-2}(q^0)^{14|r+s|-1}
\,d\Omega\,dq\cr
&\leq C\iint |D^k_{x,p}S|^2(p^0)^{2|\hat{k}|}
(p^0)^{4|r+s|-2}(q^0)^{14|r+s|-1}
\,d\Omega\,dq,\label{2.19}
\end{align}
where we used $-p_0\leq Cp'_0q'_0$ and $p^0\leq Cp'^0q'^0$.
For the case $N=\infty$,
\begin{align}
&\hspace{-0.5cm}\iint |D^k_{x,p}S|^2
h^2(p)h^{-2}(p')h^{-2}(q')(p^0)^{2(i+j)+2|\overline{r+s}|}
(q^0)^{8(i+j)+6|\overline{r+s}|}\,d\Omega\,
\frac{1}{p_0'q_0'p_0q_0}\,dq\cr
&\leq C\iint |D^k_{x,p}S|^2h^{-2}(q)(p^0)^{4|r+s|-2}(q^0)^{14|r+s|-1}
\,d\Omega\,dq,\label{2.20}
\end{align}
where we used $p'^0+q'^0=p^0+q^0$ and $-p_0\leq Cp'_0q'_0$.
The boundedness of the integrals \eqref{2.19} and \eqref{2.20} are the 
conditions on $S$ for the gain term. The arguments for the loss term are much 
simpler, and it is easily shown that the condition for the gain term
implies that for the loss term.
Consequently, we obtain the following lemma.
\begin{lemma}\label{Sbound}{\rm\cite{BCB73}}
Let $\mu\geq 5$, and suppose that the collision cross-section $S$ satisfies 
the following conditions for multi-indices $k+r=l$ satisfying $|l|\leq\mu$.
\begin{align*}
\iint |D^k_{x,p}S|^2(q^0)^{14|r|-1}\,d\Omega\,dq\leq C(p^0)^{2-4|r|-2|\hat{k}|}\quad
&\mbox{for}\quad N<\infty,\\
\iint |D^k_{x,p}S|^2 e^{-2q^0}(q^0)^{14|r|-1}\,d\Omega\,dq\leq C(p^0)^{2-4|r|}\quad
&\mbox{for}\quad N=\infty.
\end{align*}
Then, the collision operator $Q$ satisfies $\mu-N$ regularity.
\end{lemma}
\begin{remark}
Note that the condition on $S$ for finite $N$ is stronger than that
for the case $N=\infty$, and that the former condition implies the latter one.
\begin{align*}
&\hspace{-0.5cm}\iint |D^k_{x,p}S|^2 e^{-2q^0}(q^0)^{14|r|-1}\,d\Omega\,dq\\
&\leq\iint |D^k_{x,p}S|^2(q^0)^{14|r|-1}\,d\Omega\,dq
\leq C(p^0)^{2-4|r|-2|\hat{k}|}
\leq C(p^0)^{2-4|r|}.
\end{align*}
In other words, if the collision cross-section $S$ satisfies the condition
for finite $N$, then it also satisfies the condition for the case $N=\infty$. 
\end{remark}

\section{The non-negativity problem for the Einstein-Boltzmann system}
\setcounter{equation}{0}
We now consider the EB system and the question of the non-negativity of
its solutions for non-negative initial data. The Cauchy problem for the EB 
system was first studied by Bancel and Choquet-Bruhat \cite{B73,BCB73}.
We apply Lemma 2.5 to the results of \cite{B73,BCB73} to obtain
the following existence theorem.
\begin{theorem}{\rm\cite{B73,BCB73}}
Suppose that Cauchy data are prescribed as
\[
g_{\alpha\beta}(0)\in H_{\mu+1}(\omega_0),\quad
\partial_0 g_{\alpha\beta}(0)\in H_{\mu}(\omega_0),\quad
f(0)\in H_{\mu,N}(\hat{\omega}_0),
\]
such that
\[
|g_{\alpha\beta}(0)-\eta_{\alpha\beta}|\leq\varepsilon-\delta,\quad\delta>0.
\]
If $\mu\geq 5$, $N\geq 6$, and the collision cross-section $S$ satisfies the 
conditions of Lemma 2.5,
then there exist a domain $\Omega$ in $\bbr^4$, which admits $\omega_0$ as a 
Cauchy surface,
and a function $f$ on $\hat{\Omega}$ such that
\begin{itemize}
\item[1.] $g_{\alpha\beta}\in H_{\mu+1}(\Omega)$ and $f\in H_{\mu,N}(\hat{\Omega})$.
\item[2.] $g_{\alpha\beta}$ satisfies 
$|g_{\alpha\beta}-\eta_{\alpha\beta}|\leq\varepsilon$ on $\Omega$.
\item[3.] $g_{\alpha\beta}$ and $f$ satisfy the EB system.
\item[4.] $g_{\alpha\beta}$ and $f$ induce the prescribed Cauchy data
on $\omega$ and $\hat{\omega}$ respectively.
\end{itemize}
This solution is unique in $\Omega$ and depends continuously on the Cauchy data.
\end{theorem}

In this section we consider the non-negativity of the distribution function
constructed in the above theorem. Note that we already have a solution 
$(f,g_{\alpha\beta},\Omega)$ of the EB system, while the non-negativity problem 
concerns only the distribution function $f$. Hence, it is enough to consider 
the Boltzmann equation on a given curved spacetime and show non-negativity in 
this case. Conditions for the non-negativity of solutions of the Boltzmann
equation on a given spacetime have been given by Bichteler \cite{B67} and 
Tadmon \cite{T10}. We will investigate to what extent the arguments for 
non-negativity given in \cite{B67,T10} apply to the solutions of the EB 
system constructed in \cite{B73,BCB73}.

Theorem 3.1 holds for any weighted spaces with $N\geq 6$, but we will only 
consider the case $N=\infty$, i.e., the exponentially bounded case.
Firstly, any functions that have polynomial decay rates can be approximated by
functions having exponential decay rates. Hence, non-negativity of the case
$N=\infty$ implies the corresponding result in the case of finite $N$ by 
suitable approximations. Secondly, as we checked in Remark 2.2, a collision 
cross-section which satisfies the condition for the case of finite $N$ 
also satisfies the condition for the case $N=\infty$. So, any property that 
holds under the condition for the case $N=\infty$ holds under the condition
for the case of finite $N$.

\subsection{The result of Tadmon}
The original method for the non-negativity problem introduced by Lu and Zhang 
\cite{LZ01} has recently been applied to the general relativistic case by 
Tadmon \cite{T10}. A curved spacetime was assumed to be given, and mild 
solutions of the Boltzmann equation were considered, which are defined as 
follows. Consider the vector field 
$(p^\alpha,-\Gamma^\alpha_{\beta\gamma}p^\beta p^\gamma)$ 
as in \cite{B67}, and parametrize its integral curves by $x^0=s$ on the mass 
hyperboloid. Let $(X(s),P(s))\in\bbr^3\times\bbr^3$ denote an integral curve 
with $X^i(0)=x^i$ and $P(0)=p$ parametrized by $s$. This has the physical 
interpretation of a particle path. Along this curve, the 
Boltzmann equation is written as follows.
\begin{equation}\label{3.5}
f(t,X(t),P(t))=f(0)+\int_0^t K(f)(s,X(s),P(s))\,ds,
\end{equation}
where $f(0)$ is an initial datum evaluated at $(x^i,p)$, and $K$ denotes
\[
K(f)(x,p)=\frac{1}{p^0}Q(f,f)(x,p).
\]
Mild solutions are now defined as follows:
{\it a function $f$ is called a mild solution of the Boltzmann equation with 
measurable initial value $f(0)$,
if $f$ is measurable, $K(f)$ is $L^1_{{\rm loc}}(\bbr_+)$ along the integral 
curves, and \eqref{3.5} holds.}

For the rest of this section, $x$ will denote three-dimensional vector
consisting of the spatial components $x^i$ where 
\[
x^\alpha=(x^0,x^i)=(t,x)\in\bbr_+\times\bbr^3,
\]
and the distribution function will be written as $f(t,x,p)$ instead of $f(x,p)$.
The main theorem of \cite{T10} can be stated as follows.
\begin{theorem}
Let $f$ be a mild solution of the Boltzmann equation with a non-negative 
initial datum $f(0)$. Assume the following conditions.
\begin{align*}
\mbox{\rm(i)}\quad&\frac{S(t,x,p,q,\Omega)}{p^0q_0}
=\frac{S(t,x,q,p,\Omega)}{q^0p_0},\\
\mbox{\rm(ii)}\quad&\frac{S(t,x,p',q',\Omega)}{p'^0q'_0}
=\frac{\partial(p,q)}{\partial(p',q')}\frac{S(t,x,p,q,\Omega)}{p^0q_0},\\
\mbox{\rm(iii)}\quad&{\rm ess}\sup_{x,p}\iint_{\bbr^3\times S^2}
\frac{S\, q^0}{(p^0)^2q_0}|f(t,x,q)||g|^{\frac{1}{2}}\,d\Omega\,dq
=:\alpha(t)\in L^1_{{\rm loc}}(\bbr_+),\\
\mbox{\rm(iv)}\quad&{\rm ess}\sup_{y,v}
\frac{|l(t,y,v)|J(y,v)}{v^0}=:\beta(t)\in L^1_{{\rm loc}}(\bbr_+),
\end{align*}
where $l$ is defined by
\[
l(s,X(s),P(s)):=\partial_s\Big[P^0(s)J^{-1}(X(s),P(s))\Big]
\]
with $J$, the Jacobian for
\[
dx\,dp=\left|\frac{\partial(x,p)}{\partial(X(s),P(s))}\right|\,dX(s)\,dP(s)
=:J(X(s),P(s))\,dX(s)\,dP(s).
\]
Then, $f$ is non-negative.
\end{theorem}
\begin{remark}
The conditions on $\alpha(t)$ and $\beta(t)$ should be replaced by
$\alpha,\beta\in L^1_{{\rm loc}}(I)$ for some finite interval $I$, because the 
solutions in Theorem 3.1 are local in time.
Hence, it is enough to show the boundedness of $\alpha$ and $\beta$ on $I$.
\end{remark}

In this part, we will show that, with an extra assumption, the solution of 
Theorem 3.1 satisfies the conditions (i)--(iv). Since, as explained below,
the extra condition can be arranged by making a coordinate change, 
non-negativity for the EB case is ensured. We remark that the strong solutions 
in Theorem 3.1 satisfy the definition of mild solutions.\bigskip

\noindent{\bf Non-negativity of $f$.}
We consider the conditions (i)--(iv) separately in the following.\bigskip

\noindent{\bf (i)--(ii)}
The collision cross-section $S$ is given by \eqref{2.1}, which implies
\[
S(t,x,p,q,\Omega)=S(t,x,q,p,\Omega)=S(t,x,p',q',\Omega),
\]
and it is well known that the Jacobian in (ii) is given by
\[
\frac{\partial(p,q)}{\partial(p',q')}=\frac{p_0q_0}{p'_0q'_0}.
\]
Hence, the conditions (i)--(ii) can be written as
$p_0/p^0=q_0/q^0=p'_0/p'^0$ for any $p$ and $q$.
However, this leads to
\begin{equation}\label{3.6}
g_{0i}=0,
\end{equation}
which is not assumed in Theorem 3.1 and is the extra condition referred to
above.\bigskip

\noindent{\bf (iii)} For the metric $g_{\alpha\beta}$ given in Theorem 3.1,
its determinant $|g|$ is bounded, and $q_0$ and $q^0$ are equivalent.
Therefore, (iii) is estimated as follows.
\begin{align*}
\alpha(t)&\leq C\frac{1}{(p^0)^2}\iint S\,|f(t,x,q)|\,d\Omega\,dq\\
&\leq C\frac{1}{(p^0)^2}\left(\iint S^2 e^{-2q^0}\,d\Omega\,dq\right)^{\frac{1}{2}}
\left(\iint e^{2q^0}|f(t,x,q)|^2\,d\Omega\,dq\right)^{\frac{1}{2}},
\end{align*}
where $e^{q^0}$ is the weight function for $N=\infty$ case.
Since $f\in H_{\mu,N}(\hat{\Omega})$ with $\mu\geq 5$, the second integral 
above is bounded.
The condition on $S$ given in Lemma 2.5 implies the boundedness of the first 
integral, and this shows that the condition (iii) holds in Theorem 3.1.\bigskip

\noindent{\bf (iv)} As for the last condition, we recall that the integral 
curve $(X(s),P(s))$
is defined as follows. Let $X(s)$ and $P(s)$ denote for each $x$ and $p$,
\[
X(s)=X(s,0,x,p),\quad P(s)=P(s,0,x,p),
\]
and they satisfy
\begin{equation}\label{3.7}
\left\{
\begin{aligned}
&\partial_sX^i(s)=\frac{P^i(s)}{P^0(s)},\quad X^i(0)=x^i,\\
&\partial_sP^i(s)=-\Gamma^i_{\alpha\beta}
\frac{P^\alpha(s)P^\beta(s)}{P^0(s)},\quad P^i(0)=p^i,
\end{aligned}
\right.
\end{equation}
where the Christoffel symbols are evaluated at $X^\alpha(s)=(s,X(s))$.
We first consider the quantity
\begin{align*}
\partial_s\Big[P^0(s)\Big]
&=(\partial_{x^\alpha}P^0)(s)(\partial_sX^\alpha(s))
+(\partial_{p^i}P^0)(s)(\partial_sP^i(s))\\
&=-\frac{P^\beta(s)P^\gamma(s)}{2P_0(s)}(\partial_\alpha g_{\beta\gamma})
\frac{P^\alpha(s)}{P^0(s)}
+\frac{P_i(s)}{P_0(s)}\Gamma^i_{\alpha\beta}\frac{P^\alpha(s)P^\beta(s)}{P^0(s)},
\end{align*}
where we used 
$\partial_{x^\alpha}p^0=-p^\beta p^\gamma\partial_\alpha g_{\beta\gamma}/(2p_0)$
and $\partial_{p^k}p^0=-p_k/p_0$ as in the proof of Lemma 2.2 with the 
equations \eqref{3.7}.
Since $g_{\alpha\beta}\in H_{\mu+1}(\Omega)$ with $\mu\geq 5$,
$\partial_\alpha g_{\beta\gamma}$ and $\Gamma^i_{\alpha\beta}$ are bounded.
Lemma 2.1 then gives the estimate
\begin{equation}\label{3.8}
\partial_s\Big[P^0(s)\Big]
\leq C P^0(s).
\end{equation}

Consider now $\partial(X(s),P(s))/\partial (x,p)$.
We need the following calculations.
Differentiate the first equation of \eqref{3.7} with respect to $x^j$ to obtain
\begin{align*}
\partial_s\partial_{x^j}X^i(s)
&=\frac{\partial_{x^j}P^i(s)}{P^0(s)}-\frac{P^i(s)}{(P^0(s))^2}
\partial_{x^j}\Big[P^0(s)\Big]\\
&=\frac{\partial_{x^j}P^i(s)}{P^0(s)}-\frac{P^i(s)}{(P^0(s))^2}
\left(-\frac{P^\beta(s)P^\gamma(s)}{2P_0(s)}(\partial_\alpha g_{\beta\gamma})
(\partial_{x^j}X^\alpha(s))-\frac{P_i(s)}{P_0(s)}(\partial_{x^j}P^i(s))\right),
\end{align*}
and as a consequence we obtain the estimate
\begin{equation}\label{3.9}
\partial_s|\partial_{x}X(s)|
\leq C|\partial_xX(s)|+C\frac{1}{P^0(s)}|\partial_xP(s)|.
\end{equation}
Similarly we have
\begin{equation}\label{3.10}
\partial_s|\partial_{p}X(s)|
\leq C|\partial_pX(s)|+C\frac{1}{P^0(s)}|\partial_pP(s)|.
\end{equation}
Differentiate the second equation of \eqref{3.7} with respect to $x^j$ to 
obtain
\begin{align*}
\partial_s\partial_{x^j}P^i(s)
&=-(\partial_\gamma\Gamma^i_{\alpha\beta})(\partial_{x^j}X^\gamma(s))
\frac{P^\alpha(s)P^\beta(s)}{P^0(s)}
-2\Gamma^i_{\alpha\beta}\frac{P^\beta(s)}{P^0(s)}\partial_{x^j}\Big[P^\alpha(s)\Big]\\
&\quad +\Gamma^i_{\alpha\beta}\frac{P^\alpha(s)P^\beta(s)}{(P^0(s))^2}
\partial_{x^j}\Big[P^0(s)\Big].
\end{align*}
By similar arguments, we have
\begin{equation}\label{3.11}
\partial_s|\partial_xP(s)|\leq CP^0(s)|\partial_xX(s)|+C|\partial_xP(s)|,
\end{equation}
and similarly again
\begin{equation}\label{3.12}
\partial_s|\partial_pP(s)|\leq CP^0(s)|\partial_pX(s)|+C|\partial_pP(s)|.
\end{equation}
Due to \eqref{3.8}, the inequalities \eqref{3.9}--\eqref{3.12} can be combined 
as follows.
\begin{align*}
\partial_s\Big[|P^0(s)\partial_xX(s)|+|\partial_xP(s)|\Big]
&\leq C\Big(|P^0(s)\partial_xX(s)|+|\partial_xP(s)|\Big),\\
\partial_s\Big[|P^0(s)\partial_pX(s)|+|\partial_pP(s)|\Big]
&\leq C\Big(|P^0(s)\partial_pX(s)|+|\partial_pP(s)|\Big).
\end{align*}
We now use Gr{\"o}nwall's lemma with the conditions $X(0)=x$ and $P(0)=p$ to 
obtain
\begin{align*}
|P^0(s)\partial_xX(s)|+|\partial_xP(s)|
&\leq C\Big(|P^0(0)\partial_xX(0)|+|\partial_xP(0)|\Big)\leq Cp^0,\\
|P^0(s)\partial_pX(s)|+|\partial_pP(s)|
&\leq C\Big(|P^0(0)\partial_pX(0)|+|\partial_pP(0)|\Big)\leq C,
\end{align*}
which imply the estimates
\begin{equation}\label{3.13}
\begin{aligned}
&|\partial_xX(s)|\leq C\frac{p^0}{P^0(s)},\quad
|\partial_pX(s)|\leq C\frac{1}{P^0(s)},\\
&|\partial_xP(s)|\leq Cp^0,\quad
|\partial_pP(s)|\leq C.
\end{aligned}
\end{equation}
If we apply \eqref{3.13} to \eqref{3.9}--\eqref{3.12} again,
we can see that $s$-derivatives of the above quantities satisfy the same 
estimates.
\begin{equation}\label{3.14}
\begin{aligned}
&\partial_s|\partial_xX(s)|\leq C\frac{p^0}{P^0(s)},\quad
\partial_s|\partial_pX(s)|\leq C\frac{1}{P^0(s)},\\
&\partial_s|\partial_xP(s)|\leq Cp^0,\quad
\partial_s|\partial_pP(s)|\leq C
\end{aligned}
\end{equation}

On the other hand, we can consider the following integral curves.
Let $Y(s)$ and $V(s)$ denote for each $y$ and $v$,
\[
Y(s)=Y(s,t,y,v),\quad V(s)=V(s,t,y,v),
\]
and they satisfy
\begin{equation*}
\left\{
\begin{aligned}
&\partial_sY^i(s)=\frac{V^i(s)}{V^0(s)},\quad Y^i(t)=y^i,\\
&\partial_sV^i(s)=-\Gamma^i_{\alpha\beta}
\frac{V^\alpha(s)V^\beta(s)}{V^0(s)},\quad V^i(t)=v^i,
\end{aligned}
\right.
\end{equation*}
where the Christoffel symbols are evaluated at $Y^\alpha(s)=(s,Y(s))$.
Then, by the same arguments as above, we obtain
\begin{align*}
\partial_s\Big[|V^0(s)\partial_yY(s)|+|\partial_yV(s)|\Big]
&\leq C\Big(|V^0(s)\partial_yY(s)|+|\partial_yV(s)|\Big),\\
\partial_s\Big[|V^0(s)\partial_vY(s)|+|\partial_vV(s)|\Big]
&\leq C\Big(|V^0(s)\partial_vY(s)|+|\partial_vV(s)|\Big),
\end{align*}
and use Gr{\"o}nwall's lemma with $Y(t)=y$ and $V(t)=v$,
\begin{align*}
|V^0(s)\partial_yY(s)|+|\partial_yV(s)|
&\leq C\Big(|V^0(t)\partial_yY(t)|+|\partial_yV(t)|\Big)\leq Cv^0,\\
|V^0(s)\partial_vY(s)|+|\partial_vV(s)|
&\leq C\Big(|V^0(t)\partial_vY(t)|+|\partial_vV(t)|\Big)\leq C,
\end{align*}
and as a consequence the following estimates are obtained.
\begin{equation}\label{3.15}
\begin{aligned}
&|\partial_yY(s)|\leq C\frac{v^0}{V^0(s)},\quad
|\partial_vY(s)|\leq C\frac{1}{V^0(s)},\\
&|\partial_yV(s)|\leq Cv^0,\quad
|\partial_vV(s)|\leq C.
\end{aligned}
\end{equation}

We now consider the quantity
\[
l(s,X(s),P(s))=\partial_s\Big[P^0(s)\Big]J^{-1}(X(s),P(s))
+P^0(s)\partial_s\Big[J^{-1}(X(s),P(s))\Big],
\]
from which we have for $y=X(t)$ and $v=P(t)$,
\begin{equation}\label{3.16}
l(t,y,v)J(y,v)=\partial_s\Big[P^0(s)\Big]_{s=t}
+v^0J(y,v)\partial_s\Big[J^{-1}(X(s),P(s))\Big]_{s=t}.
\end{equation}
By applying \eqref{3.8} to the first quantity above, we obtain
\begin{equation}\label{3.17}
\left|\partial_s\Big[P^0(s)\Big]_{s=t}\right|\leq Cv^0.
\end{equation}
Note that $J(y,v)$ is written as 
\[
J(y,v)=J(X(s),P(s))_{s=t}=\left|
\frac{\partial(x,p)}{\partial(X(s),P(s))}\right|_{s=t}
=\left|\frac{\partial(Y(s),V(s))}{\partial(y,v)}\right|_{s=0},
\]
where we used
\[
x=Y(0,t,X(t),P(t))=Y(0,t,y,v)=Y(0)\quad\mbox{{\rm and similarly}}\quad
p=V(0).
\]
Thanks to the multilinearity of determinant, $J(y,v)$ is estimated by 
\eqref{3.15} as follows.
\begin{equation}\label{3.18}
J(y,v)\leq C\left(\frac{v^0}{V^0(s)}\right)^3_{s=0}
=C\left(\frac{v^0}{p^0}\right)^3.
\end{equation}
On the other hand, as for $J^{-1}(X(s),P(s))$, we have the same estimates for
$\partial_xX(s)$ and $\partial_s\partial_xX(s)$,
$\partial_pX(s)$ and $\partial_s\partial_pX(s)$, and so on, as in \eqref{3.13} 
and \eqref{3.14}.
Thanks to the multilinearity again, the following estimate is obtained.
\begin{equation}\label{3.19}
\left|\partial_s\Big[J^{-1}(X(s),P(s))\Big]\right|_{s=t}\leq
C\left(\frac{p^0}{P^0(s)}\right)^3_{s=t}=C\left(\frac{p^0}{v^0}\right)^3.
\end{equation}
We combine \eqref{3.17}--\eqref{3.19} and \eqref{3.16} to obtain
\[
|l(t,y,v)|J(y,v)\leq Cv^0.
\]
This gives boundedness of $\beta$ on $I$, and therefore it is proved that the 
condition (iv)
holds in Theorem 3.1.

It remains to discuss the extra condition $g_{0i}=0$ which is required to 
ensure that the conditions of Tadmon's theorem are satisfied. We are concerned 
here with a solution on a local region, where the components of the metric 
in a certain coordinate system are close to those of the Minkowski metric.
Let the original coordinates be $(t,x^i)$. Now choose new coordinates 
$(t,\tilde x^i)$ with the properties that $\tilde x^i$ agrees with $x^i$ for
$t=0$ and the $x^i$ are constant along the integral curves of the unit normal 
vector field to the family of hypersurfaces of constant $t$. Then the 
components of the metric in the new coordinate system satisfy the desired 
condition. The new coordinates might not be defined on exactly the same region
as the old ones so that non-negativity is obtained on a slightly smaller domain.
Given the fact that the non-negativity theorem being discussed here is local
in nature, this is not a major disadvantage.    

\subsection{The result of Bichteler and a new approach}
In Bichteler's paper \cite{B67}, it was proved that under certain hypotheses
a local solution of the Boltzmann equation exists on a given curved spacetime, 
and that it is non-negative. A four-dimensional manifold $M$ and a Lorentzian 
metric $g_{\alpha\beta}$ were assumed to be given, and the Boltzmann equation was 
written in the form
\begin{equation}\label{3.1}
\mathcal{L}_Xf=\iiint W(12\to 34)\,\delta(1+2-3-4)
\Big(f(3)f(4)-f(1)f(2)\Big)\,d2\,d3\,d4,
\end{equation}
where $1$, $2$, $3$, and $4$ stand for $p^\alpha$, $q^\alpha$, $p'^\alpha$, and 
$q'^\alpha$ respectively.
The left hand side of \eqref{3.1} means that the derivative of $f$ along the 
integral curves of the vector field
$X=(p^\alpha,-\Gamma^\alpha_{\beta\gamma}p^\beta p^\gamma)$
exists at almost all points of the domain of $f$.
The transition rate $W$ is given as
\begin{equation*}
W(12\to 34)=k\sigma(12\to 34),
\end{equation*}
where $k$ is some kinematical factor, and $\sigma$ can be regarded as the same 
quantity
as the scattering kernel defined in \eqref{2.1}.
The volume forms $d2$, $d3$, and $d4$ are same as in our case.
\[
d2=\frac{|g|^{\frac{1}{2}}}{-q_0}\,dq,
\]
and $d3$ and $d4$ are similarly defined.
The main theorem of \cite{B67} is the following.
\begin{theorem}{\rm\cite{B67}}
Let $f(0)$ be a measurable function on $\hat{\omega}_0$,
and suppose that there exists a continuous timelike vector field
$\beta^\alpha(x)$ on $\omega_0$ such that
\[
f(x,p)\leq C e^{\beta_\alpha(x)p^\alpha}\quad\mbox{{\rm on}}\quad \hat{\omega}_0,
\]
and the scattering kernel satisfies the following property.
\begin{equation}\label{3.2}
\iint \sigma(12\to 34)\,\delta(1+2-3-4)\,d3\,d4\leq {\rm const.}
\end{equation}
Then, there exists a solution $f$ to the Boltzmann equation \eqref{3.1} on a 
domain $\hat{\Omega}$,
which is again exponentially bounded.
Moreover, if $f(0)\geq 0$, then also $f\geq 0$.
\end{theorem}

The condition \eqref{3.2} can be stated in our notation as 
\begin{equation}\label{3.3}
\int_{S^2}\sigma(\varrho,\theta)\,d\Omega\leq {\rm const.}
\end{equation}
The corresponding conditions on the scattering kernel for the EB case are 
given in Lemma 2.5
with the relation \eqref{2.1} between the scattering kernel and the collision 
cross-section. Although the conditions in Lemma 2.5 are much more complicated 
than \eqref{3.3}, we can see that they do not imply \eqref{3.3}. Hence 
Theorem 3.2 cannot be applied directly to prove Theorem 3.1.

If there exists a solution of the Boltzmann equation with non-negative
initial data which becomes negative at some point in phase space then it is 
natural to follow the particle path through that point backwards in time.
This curve must meet the Cauchy surface and thus there is at least
one point on it where $f$ is non-negative. Hence there exists a point  
$(t_*,x_*,p_*)$ on the curve where it is zero and immediately afterwards 
negative. At this point the derivative of $f$ in the direction of the particle 
path is non-negative because the loss term vanishes. It might however be zero. 
This type of consideration plays a role in the proof of positivity in 
\cite{B67}

A small modification of this idea can be used to give a relatively simple
proof of non-negativity for the EB system. We first introduce the following 
notation,
\[
\mathcal{L}=\frac{\partial}{\partial t}+\frac{p^i}{p^0}
\frac{\partial}{\partial x^i}
-\Gamma^i_{\alpha\beta}\frac{p^\alpha p^\beta}{p^0}\frac{\partial}{\partial p^i},
\]
and write the Boltzmann equation as $\mathcal{L}f=(p^0)^{-1}Q(f,f)$.
We then modify the Boltzmann equation by adding a small quantity on the right 
hand side.
\begin{equation}
\mathcal{L}f_\eta=\frac{1}{p^0}Q(f_\eta,f_\eta)+\eta e^{-|p|^2},
\end{equation}
where $\eta>0$ is a small parameter.
The above equation is understood to be defined on $\hat{\Omega}$ with
a metric $g_{\alpha\beta}$ constructed as in Theorem 3.1.
Since the quantity $\eta e^{-|p|^2}$ is sufficiently smooth and square 
integrable on $\hat{\Omega}$,
existence of solutions is easily proved by the same argument of \cite{BCB73}
using the following lemma.
\begin{lemma}{\rm\cite{B73}}
Let $g_{\alpha\beta}$ and $\hat{\Omega}$ be the metric and the domain given in 
Theorem 3.1,
and consider the Cauchy problem for $\mathcal{L}f=g$ with $f(0)$ defined on 
$\hat{\omega}_0$.
Then, for any $0\leq s\leq t$, there exists a solution in 
$H_{\mu,N}(\hat{\Omega})$ satisfying the energy inequality,
\[
\|f(s)\|^2_{H_{\mu,N}(\hat{\omega}_s)}\leq C\left(
\|f(0)\|^2_{H_{\mu,N}(\hat{\omega}_0)}
+\int_0^t\|g(\tau)\|^2_{H_{\mu,N}(\hat{\omega}_\tau)}\,d\tau\right),
\]
where $C$ is a constant depending only on $g_{\alpha\beta}$, $\Omega$, and $\mu$.
\end{lemma}

As a result, we have two distribution functions $f$ and $f_\eta$ in a common 
domain $\hat{\Omega}$, and we can see that $f_\eta$ is non-negative by 
reasoning as above. For in the case of the modified equation the derivative
of $f$ along the particle path is strictly positive at the point 
$(t_*,x_*,p_*)$, a contradiction. Non-negativity of $f$ is now proved by 
showing continuous dependence on $\eta$ for $f_\eta$.\bigskip

\noindent{\bf Non-negativity of $f$.}
We subtract the modified equation from the original one and write $F=f-f_\eta$
to get
\[
\mathcal{L}F=\frac{1}{p^0}Q(f,f)-\frac{1}{p^0}Q(f_\eta,f_\eta)-\eta e^{-|p|^2}.
\]
By applying Lemma 3.1, we obtain the following integral inequality.
\[
\|F(t)\|^2_{H_{\mu,N}(\hat{\omega}_t)}\leq C\left(\eta+\int_0^t
\left\|\frac{1}{p^0}Q(f,f)(\tau)
-\frac{1}{p^0}Q(f_\eta,f_\eta)(\tau)\right\|^2_{H_{\mu,N}(\hat{\omega}_\tau)}
\,d\tau\right),
\]
where we used $F(0)=f(0)-f_\eta(0)\equiv 0$.
From the bilinearity of $Q$, the above difference of collision terms can be 
written 
\[
Q(f,f)-Q(f_\eta,f_\eta)=Q(f,F)+Q(F,f_\eta).
\]
We now apply Lemma 2.5 to obtain
\[
\|F(t)\|^2_{H_{\mu,N}(\hat{\omega}_t)}\leq C\left(\eta
+\int_0^t\|F(\tau)\|^2_{H_{\mu,N}(\hat{\omega}_\tau)}
\,d\tau\right),
\]
which implies that 
$\displaystyle\lim_{\eta\to 0}\|F(t)\|_{H_{\mu,N}(\hat{\omega}_t)}=0$, and consequently
$\displaystyle\lim_{\eta\to 0}\|F(t)\|_{H_{\mu,N}(\hat{\Omega})}=0$.
Since $\mu\geq 5$, we obtain
\[
\lim_{\eta\to 0}\|f-f_\eta\|_{L^\infty(\hat{\Omega})}=0,
\]
and this proves non-negativity of $f$.

The above argument actually proves a stronger statement. Consider a situation
where the initial data are not everywhere non-negative. If $(t,x,p)$ is a
point of phase space such the distribution function is non-negative at the 
point where the particle path through $(t,x,p)$ meets the Cauchy surface then 
we can argue as before to get a contradiction. Thus a statement is obtained 
about the non-negativity of $f$ at certain points of phase space.

\section{Further considerations}

An important concept in general relativity is that of general covariance.
In the context of mathematical relativity this means that it should be 
possible to express the Einstein-matter equations in a form which is 
invariant under diffeomorphisms. This assertion applies in particular to
the Einstein-Boltzmann system. The results of the previous sections are 
concerned with solutions of the EB system expressed in a local coordinate 
system. In other words they are results about the reduced Einstein equations 
in a harmonic coordinate system. These can be used to establish properties of 
the Einstein-Boltzmann sytem which are diffeomorphism invariant by 
standard procedures (cf. \cite{R}, Chapter 9). From this point of view
it is natural to prove a global positivity result for the EB system
and for this it is enough to show that a solution of the Boltzmann equation
on a globally hyperbolic spacetime with non-negative initial data is
non-negative. A problem in formulating a theorem of this kind is to 
identify a suitable class of collision cross-sections.  

The conditions on the collision cross-section used in \cite{BCB73} and 
in the previous sections are conditions on the function $S$ and, as such,
are coordinate dependent. For a global theorem it is necessary to formulate
an invariant condition and the function $S$ is not invariant since it
depends explicitly on the spacetime coordinates. A better alternative 
would be to formulate a condition on $\sigma$. In the formulation of 
the reduced EB system using a coordinate basis the function $\sigma$ 
still depends on the metric components. If instead the EB system is 
formulated in an orthonormal frame in the way explained in the introduction
this dependence on the metric is eliminated. This formulation also has
an advantage for the consideration of the question of what collision 
cross-sections are physically reasonable. When the orthonormal frame approach
is used the scattering kernel in general relativity is identical to that in
special relativity. Thus the problem of identifying physically reasonable
scattering kernels is reduced to the corresponding problem in special 
relativity which is better understood. 

What has just been said provides a strong motivation for considering 
the reduced system in the orthonormal frame formulation in more detail.
The method of proof of the theorems of \cite{BCB73} extends easily to 
this case. The inequalities on the derivatives $D^k_{x,p}S$ in Lemma 
\ref{Sbound} are replaced by the analogous estimates for $D^k_vS$
which are actually somewhat simpler. In fact there is one estimate which is 
necessary for those proofs which is not included in \cite{BCB73}. In the 
iteration used in that paper there is a new metric in each step of the 
iteration. In estimating the differences of iterates it is in particular 
necessary to estimate the change resulting from making a change of the metric 
in the collision term. This is not mentioned in \cite{BCB73}. The orthonormal 
frame formulation eliminates this problem since in that case the collision
term does not depend on the metric components. The conditions on $D^k_vS$
which are analogous to those of Lemma \ref{Sbound} are coordinate independent 
and are thus appropriate for formulating a global theorem. Let us call these the
orthonormal frame regularity conditions. From the point of view of
comparing with physical conditions on the collision term it would be of
interest to know what conditions on $\sigma$ are required to imply the 
orthonormal frame regularity conditions. For instance, as a simple case an
assumption on the support of $\varrho$ can be made. Let us consider the 
orthonormal frame regularity conditions in more detail. The conditions of
Lemma \ref{Sbound} are replaced as follows for the case $N=\infty$. Under the 
same conditions on $\mu$, $k$, and $r$, the following inequality should hold.
\[
\iint |\partial_{\varrho}^i{\hat\sigma}(\varrho,\theta)|^2
\varrho^{2i-4|k'|}e^{-2q^0}(q^0)^{14|r|+2|k|}(p^0)^{4|r|-1}\,d\Omega\,dq\leq C,
\]
where $i\leq |k'|\leq |k|$ and, motivated by the form of the expression
$S=\lambda\varrho\sigma$, the quantity ${\hat\sigma}=\varrho\sigma$ has
been introduced. Suppose now that $\sigma$ is a smooth function and
has support contained in the set defined by $m\leq\varrho\leq M$ 
for some positive numbers $m$ and $M$. Then, the term $\varrho^{2i-4|k'|}$ is 
clearly bounded. Moreover, the condition $\varrho\leq M$ gives the following 
estimate. Since we are working in an orthonormal frame, it implies that
\[
p^0\leq \frac{1}{q^0}(q\cdot p+M')
\leq \frac{1}{q^0}(|q||p|+M')\quad\mbox{for}\quad
M'=\frac{1}{2}M^2+1.
\]
After some calculations and using the inequality $|p|\leq p^0$ we obtain 
$|p|\leq 2M'q^0$, and this implies in turn that $p^0\leq Cq^0$ for some $C$.
Hence the orthonormal frame regularity conditions hold.

Consider a solution $f$ of the Boltzmann equation on a globally hyperbolic
spacetime and suppose that the scattering kernel satisifes the orthonormal
frame regularity conditions. If the initial data are non-negative then $f$ is 
non-negative. We suppose that the statement is false and obtain a 
contradiction. If the statement is false there is a point
$(t_1,x_1,p_1)$ with $f(t_1,x_1,p_1)<0$. Let $\gamma$ be the particle path
passing through $(t_1,x_1,p_1)$. If $(t,x,p)$ is a point of $\gamma$ 
sufficiently close to the initial hypersurface then local coordinates can be 
defined such that $(t,x)$ is contained in their domain of definition and the 
conditions of the theorem of \cite{BCB73} are satisfied. Hence 
$f(t,x,p)\ge 0$ close to the initial hypersurface. Let $(t_2,x_2,p_2)$ be a 
point on $\gamma$ with 
$f(t,x,p)\ge 0$ for all points on $\gamma$ with $t\le t_2$ and $f(t,x,p)<0$ 
for some points on $\gamma$ with $t$ arbitrarily close to $t_2$. We can choose 
a local Cauchy surface through the point $(t_2,x_2)$ and a local 
coordinate system on a neighbourhood of $(t_2,x_2)$ such that the conditions 
of the theorem of \cite{BCB73} are satisfied. It then follows by the result 
stated at the end of Section 3 that $f\ge 0$ at all points on $\gamma$ with 
$t$ slightly greater than $t_2$, a contradiction.   

It should be noted that as soon as results on well-posedness of the Cauchy
problem for the EB system can be extended to wider classes of scattering
kernels the positivity result can also be extended, provided the theorem
includes a statement about continuous dependence of the solution on the 
scattering kernel in a suitable sense.
For then it suffices to approximate a scattering kernel
of the new class by a sequence of kernels of the class previously treated.   

We hope that by clarifying a number of issues the results of this paper 
will contribute to the development of a mature theory of the local and
global Cauchy problem for the Einstein-Boltzmann system in the near future.

\end{document}